\documentclass[11pt]{article}

\usepackage[a4paper,margin=1in]{geometry}
\usepackage{amsmath,amssymb,amsthm}
\usepackage{hyperref}
\usepackage{enumitem}
\usepackage{microtype}
\usepackage{xcolor}
\usepackage{subcaption}
\usepackage{float}
\usepackage{pgfplots}
\usepackage{tikz-cd}
\usepackage{needspace}
\pgfplotsset{compat=1.18}

\hypersetup{
  colorlinks=true,
  linkcolor=blue,
  citecolor=blue,
  urlcolor=blue
}

\newcommand{\ProofInAppendix}{\textit{Full proof in Appendix~\ref{app:fullproofs}.}}

\newtheorem{theorem}{Theorem}[section]
\newtheorem{lemma}[theorem]{Lemma}

\newtheorem{corollary}[theorem]{Corollary}

\theoremstyle{definition}
\newtheorem{definition}[theorem]{Definition}
\newtheorem{assumption}[theorem]{Assumption}
\newtheorem{remark}[theorem]{Remark}

\title{From Swap Axioms to Weighted Geometric Means:\\
A Characterization of AMMs}
\author{Bj\"orn Assmann\thanks{\texttt{bjoern.assmann@gmail.com}} \and Ulan Degenbaev\thanks{\texttt{udegenbaev@gmail.com}}}
\date{\today}

\begin{document}
\maketitle

\begin{abstract}

Many automated market makers can be understood through the geometry of their
\emph{trading orbits}, the sets of states reachable from one another through swaps.
In prominent designs, this geometry is captured by a simple closed-form invariant
such as the constant product $xy$ in Uniswap or a weighted geometric mean $x^w
y^{1-w}$ in Balancer.

This paper explains why these forms arise by deriving them from three basic assumptions:
\emph{validity invariance} (swaps preserve the validity of states),
\emph{Pareto efficiency} (no state on an orbit weakly dominates another),
and \emph{unit invariance} (changing measurement units does not change the mechanism).
Together, these force every trading orbit of a two-asset AMM to be a level set of a
weighted geometric mean $x^wy^{1-w}$.
Applied pairwise, the axioms extend the classification to $n$-asset pools:
orbits are level sets of $\prod_i x_i^{w_i}$ with positive weights $w_i$ summing to $1$.
Imposing token-relabeling symmetry then pins down the weights, recovering the
constant-product form $xy$ in the two-asset case and $\prod_i x_i$ in general.

The main text provides an intuitive proof sketch and discusses fees and liquidity operations.
Complete proofs and a machine-checked Lean~4 formalization accompany the paper.
\end{abstract}

%%%%%%%%%%%%%%%%%%%%%%%%%%%%%%%%%%%%%%%%%%%%%%%%%%%%%%%%%%%%%%%%%%%%%%%%%%%%%%%%
\section{Introduction}

Automated market makers (AMMs) are on-chain mechanisms that execute token swaps without an order book.
In a two-token pool, the reserve pair of token balances determines which swaps are feasible and at what prices.
The constant-product AMM popularized by Uniswap~\cite{UniswapV2} and the weighted geometric-mean AMM
used by Balancer~\cite{Balancer} are the canonical examples, but they sit inside a much larger design space of possible
reserve-based trading mechanisms.
This raises a natural question: can one derive the familiar invariant curves
from simple assumptions on the swap rule itself?

In this paper we study this question for AMMs with positive reserves and unrestricted swaps.
Rather than postulating an invariant function from the outset, we focus on the \emph{trading orbits},
which are the sets of reserve states reachable from one another by swaps.
We impose three structural requirements: swaps must preserve the validity of states (validity invariance),
no reserve state on a trading orbit may strictly dominate another coordinatewise (Pareto efficiency) and
the trading rule must be invariant under independent rescaling of the units used to measure each asset (unit invariance).

\paragraph{Two-asset classification.}
In the two-asset case, these three requirements are already highly restrictive.
Our main theorem shows that every trading orbit must be a weighted geometric-mean curve
\[
x^{w} y^{1-w} = k
\]
for some $w \in (0,1)$ and $k>0$.
The weighted geometric-mean AMMs are therefore not merely examples but forced by the axioms.
The constant-product case appears as the special case $w = 1/2$.

\paragraph{Multi-asset extension.}
The classification extends to $n$-asset AMMs. Given such an AMM, each pair $(i,j)$ of reserves defines a \emph{two-token restriction}: the intersection of a global trading orbit with the coordinate slice obtained by freezing all reserves except $x_i$ and $x_j$. If every such restriction of the AMM satisfies the two-asset axioms, then every global trading orbit is a level set of a weighted product $\prod_i x_i^{w_i}$ with positive weights $w_i$ summing to $1$. The argument reduces to showing that a subset of $\mathbb R^n$ whose every coordinate $2$-slice is an affine line of negative slope must itself be a hyperplane with positive normal.

\paragraph{Proof method.}
Our proof is geometric and elementary. The key step is to pass to log reserves $(u,v)=(\log x,\log y)$, where token-wise rescaling becomes translation. This turns the orbit structure into a translation-invariant geometry in $\mathbb R^2$, reducing the problem to a classification of additive subgroups and their cosets subject to Pareto monotonicity: the subgroup must be a line through the origin, and every orbit a parallel translate of it. Exponentiating back yields the weighted geometric-mean form. The same log-coordinate passage applies to the $n$-asset case, with hyperplanes playing the role of lines.

\paragraph{Formalization.}
We formalize the main theorem and its multi-asset extension in Lean~4. The machine-checkable proofs give confidence in our arguments and show that foundational results in DeFi mechanism design can be expressed and verified in a modern theorem prover.

\paragraph{Related work.}
Much of the CFMM literature takes a specified AMM mechanism, trading
function, or invariant as given and studies the resulting pricing, arbitrage behavior, oracle
properties, or liquidity-provider outcomes (see, e.g.,~\cite{AngerisChitraOracles2020,AngerisAgrawalEtAlMultiAsset2021,MilionisEtAlLVR2022}).
Our paper reverses this perspective.
Rather than postulating an invariant and analyzing its consequences, we ask which orbit
geometries are compatible with a swap mechanism satisfying a small number of basic
axioms.

The closest conceptual relative is Lee~\cite{Lee}, who studies general state-transition market systems
from an order-theoretic viewpoint and proves abstract existence and uniqueness results for
increasing invariants. In that framework, arbitrage-freeness is equivalent to the existence of an
increasing invariant, while completeness is equivalent to uniqueness up to monotone
reparameterization. Another paper close in spirit is~\cite{FrongilloPapireddygariWaggoner2023},
which shows that market makers satisfying axioms on valid trades admit a
potential-function representation. Our result is complementary. In the present two-asset setting, it is more explicit and
constructive. Instead of establishing the existence of some representing invariant for a broad class
of systems, we work directly with primitive swap operations and show that validity invariance,
Pareto efficiency, and unit invariance under token-wise rescaling force the trading orbits
themselves to be weighted geometric-mean curves ($x^w y^{1-w} = k$). Imposing token-relabeling
symmetry further yields the constant-product form ($xy = k$). Thus, in our setting, the invariant
family is not assumed and then analyzed, nor obtained indirectly through a general representation
theorem; it is derived directly from structural axioms by an elementary geometric argument.

Where existing Lean developments for AMMs verify properties of a given CFMM design~\cite{PuscedduBartolettiLeanAMM2024,DessalviBartolettiLluchLafuenteLeanFees2026}, our formalization verifies a structural classification result for the family of designs allowed by the axioms.

\needspace{10\baselineskip}
\paragraph{Contributions.}
In summary, the main contributions are the following:
\begin{enumerate}[nosep]
  \item State a minimal set of natural assumptions on the swap operation of a two-asset AMM.
  \item Show that these assumptions force the trading orbits to be weighted geometric-mean curves.
  \item Extend the main theorem to multi-asset AMMs.
  \item Characterize how fees and liquidity operations interact with the orbit geometry.
  \item Formalize the main theorem and the multi-asset extension in Lean~4.
\end{enumerate}

\paragraph{Overview.}
Section~\ref{sec:model} defines the minimal model and the three assumptions.
Section~\ref{sec:proofsketch} gives an intuitive, picture-driven proof sketch and extends the result to arbitrarily many assets.
Section~\ref{sec:fees_liquidity} discusses how fees and liquidity operations interact with the orbit geometry.
Section~\ref{app:lean_v2} discusses the Lean formalization.
Appendix~\ref{app:fullproofs} contains the complete proofs.

%%%%%%%%%%%%%%%%%%%%%%%%%%%%%%%%%%%%%%%%%%%%%%%%%%%%%%%%%%%%%%%%%%%%%%%%%%%%%%%%
\section{Model}
\label{sec:model}

\subsection{Definitions}

An automated market maker (AMM) pool manages reserves of two tokens, $X$ and $Y$.
Traders swap assets by depositing one token and withdrawing the other.
While real-world deployments track additional metadata such as fee accumulators, liquidity-provider shares, or tick ranges, we focus on an idealized model in which the reserve balances alone determine the state.
This abstraction isolates the geometric structure of trading orbits.

\begin{definition}[State]\label{def:state}
A \emph{state} is a pair $s=(x,y)\in\mathbb{R}^2$ representing the reserves of
tokens $X$ and $Y$ in the AMM pool.
The \emph{valid} state space is
\[
\mathcal S := \{(x,y)\in\mathbb{R}^2 \mid x>0,\; y>0\}.
\]
\end{definition}

We allow the primitive swap maps to be defined on all of $\mathbb R^2$, not just on $\mathcal S$.
This is convenient because the orbit relation below is defined as the symmetric closure of one-step swaps.
Validity will then be enforced by an axiom stating that swaps preserve validity status.

\begin{definition}[Swap primitives]\label{def:swap}\leavevmode
\begin{itemize}[leftmargin=1.6em]
\item $X(s,dx)$ denotes the post-swap state obtained by inputting $dx\ge 0$ units of token $X$ into the pool,
\item $Y(s,dy)$ denotes the post-swap state obtained by inputting $dy\ge 0$ units of token $Y$ into the pool.
\end{itemize}
Bookkeeping defines output maps $y_{\mathrm{out}}$ and $x_{\mathrm{out}}$ through
\[
X(s,dx)=(x+dx,\; y-y_{\mathrm{out}}(s,dx)),
\qquad
Y(s,dy)=(x-x_{\mathrm{out}}(s,dy),\; y+dy),
\]
where $s=(x,y)$.
Thus $y_{\mathrm{out}}(s,dx)$ is the amount of token $Y$ paid out when $dx$ units of token $X$ are input at state $s$,
and similarly for $x_{\mathrm{out}}(s,dy)$.
\end{definition}

We ignore fees in the axiomatic development that follows.
This isolates the structural constraints imposed by the axioms; the interplay of fees and liquidity operations with the orbit geometry is treated in Section~\ref{sec:fees_liquidity}.

A single swap gives a one-step state transition.
We first define adjacency by a single primitive swap and then take the transitive closure.

\begin{definition}[Adjacency]\label{def:adjacency}
For states $s,t$, write $s\leftrightarrow t$ if one can be obtained from the other
by a single primitive swap, meaning there exists $\delta\ge 0$ such that
\[
t=X(s,\delta)\quad\text{or}\quad t=Y(s,\delta)\quad\text{or}\quad
s=X(t,\delta)\quad\text{or}\quad s=Y(t,\delta).
\]
\end{definition}

This relation is symmetric by construction.

\begin{definition}[Orbit relation]\label{def:orbit}
Define $s\sim t$ if there exist states $s=s_0,s_1,\ldots,s_n=t$ such that
$s_{i-1}\leftrightarrow s_i$ for all $i=1,\ldots,n$.
For a valid state $s\in\mathcal S$, its \emph{trading orbit} is
\[
[s]:=\{t\in\mathcal S : t\sim s\}.
\]
\end{definition}

Thus two valid states lie in the same trading orbit precisely when they can be connected by a finite chain of primitive swaps.
The question of the paper is: \emph{what geometric shape can these orbits have?}

\subsection{Assumptions}

The abstract swap maps $X$ and $Y$ do not yet specify how outputs depend on reserves and inputs.
To obtain structure we impose three assumptions.

Our first assumption states that a primitive swap preserves the validity status of a state.
Equivalently, valid states remain valid under swaps, and invalid states never become valid through a single swap.

\begin{assumption}[Validity invariance]\label{ass:validity}
For every state $s\in\mathbb R^2$ and every $dx,dy\ge 0$,
\[
X(s,dx)\in\mathcal S \iff s\in\mathcal S,
\qquad
Y(s,dy)\in\mathcal S \iff s\in\mathcal S.
\]
\end{assumption}

Assumption~\ref{ass:validity} ensures that validity is constant along every adjacency edge, and hence along every orbit:
a valid state and an invalid state can never lie in the same orbit.

Our second assumption is an internal no-free-lunch condition.
Within a trading orbit, one should not be able to move to a state with weakly more of both reserves.

\begin{definition}[Pareto order]\label{def:pareto}
Write $t\succeq s$ if $t.x\ge s.x$ and $t.y\ge s.y$.
\end{definition}

\begin{assumption}[Pareto efficiency]\label{ass:pareto}
For all valid states $s,t\in\mathcal S$, if $s\sim t$ and $t\succeq s$, then $t=s$.
\end{assumption}

This condition is numeraire-free: it compares reserve vectors directly and does not refer to any external price oracle or valuation map.

Our third assumption formalizes invariance under the choice of measurement units.
Changing from ETH to wei, or from dollars to cents, rescales the recorded balances but should not alter the economic mechanism.

\begin{assumption}[Unit invariance]\label{ass:unit}
For any $a,b>0$ define $T_{a,b}(x,y)=(a x, b y)$.
For every valid state $s\in\mathcal S$ and all $dx,dy\ge 0$,
\[
X(T_{a,b}s,\,a\,dx)=T_{a,b}X(s,dx),
\qquad
Y(T_{a,b}s,\,b\,dy)=T_{a,b}Y(s,dy).
\]
In diagram form (shown for the $X$-swap):
\[
\begin{tikzcd}
s \ar[r,"T_{a{,}b}"] \ar[d,"X(\cdot{,}\,dx)"'] & T_{a{,}b}(s) \ar[d,"X(\cdot{,}\,a\,dx)"] \\
X(s{,}dx) \ar[r,"T_{a{,}b}"'] & T_{a{,}b}\bigl(X(s{,}dx)\bigr)
\end{tikzcd}
\]
\end{assumption}

In practice, on-chain arithmetic uses finite precision, so unit invariance holds only approximately.
Our model works over the reals and treats it as exact.
This isolates the structural constraints that unit invariance imposes on trading orbits.

Taken together, Assumptions~\ref{ass:validity}--\ref{ass:unit} force a remarkably rigid orbit geometry.

\begin{theorem}[Orbit classification]\label{thm:main}
Consider a two-asset AMM in the setting above,
satisfying validity invariance, Pareto efficiency, and unit invariance.
Then there exists $w\in(0,1)$ such that the trading orbits in
$\mathcal S$ are exactly the level sets of
\[
\Phi(x,y) := x^w y^{1-w}.
\]
\end{theorem}
\noindent \ProofInAppendix

\begin{remark}
Conversely, any AMM whose orbits are level sets of $x^wy^{1-w}$ for some $w\in(0,1)$ satisfies the three axioms. Hence Theorem~\ref{thm:main} is in fact a characterization.
\end{remark}

The weight $w$ reflects the relative role of the two assets and is left free
by the axioms.
If we further require that relabeling the two tokens does not change the
orbit structure, the weight is pinned down.

\begin{corollary}[Token symmetry forces constant product]\label{cor:main_cp}
Under the hypotheses of Theorem~\ref{thm:main}, additionally assume
\emph{token relabeling symmetry}: swapping the labels of $X$ and $Y$
preserves the orbit relation.
Then $w=\tfrac12$ and every orbit satisfies $xy=\mathrm{const}$.
\end{corollary}
\noindent \ProofInAppendix

The next section explains why the assumptions lead to weighted geometric-mean curves, and extends the classification to multi-asset pools.
%%%%%%%%%%%%%%%%%%%%%%%%%%%%%%%%%%%%%%%%%%%%%%%%%%%%%%%%%%%%%%%%%%%%%%%%%%%%%%%%
\section{Derivation of Weighted Geometric Means}
\label{sec:proofsketch}

This section presents a high-level derivation of the main result: under the three assumptions of Section~\ref{sec:model}, the trading orbits of a two-asset AMM are exactly the level sets of a weighted geometric mean $x^w y^{1-w}$. Section~\ref{sec:multiasset} extends the classification to multi-asset pools.
We emphasize the key ideas and geometric intuition, deliberately omitting technical details in favor of readability.
Readers seeking complete proofs should consult Appendix~\ref{app:fullproofs}, while those interested in full machine-checked rigor can find the Lean formalization in Section~\ref{app:lean_v2}.

We restate the theorem here for convenience.

\begin{theorem}[Orbit classification (Theorem~\ref{thm:main} restated)]
Consider a two-asset AMM in the setting of Section~\ref{sec:model},
satisfying validity invariance, Pareto efficiency, and unit invariance.
Then there exists $w\in(0,1)$ such that the trading orbits in
$\mathcal S$ are exactly the level sets of
\[
\Phi(x,y) := x^w y^{1-w}.
\]
\end{theorem}

We first analyze the orbit through $(1,1)$ in four steps
(Figure~\ref{fig:derivation}), and then extend the argument to arbitrary orbits.

\subsection{The orbit through $(1,1)$}

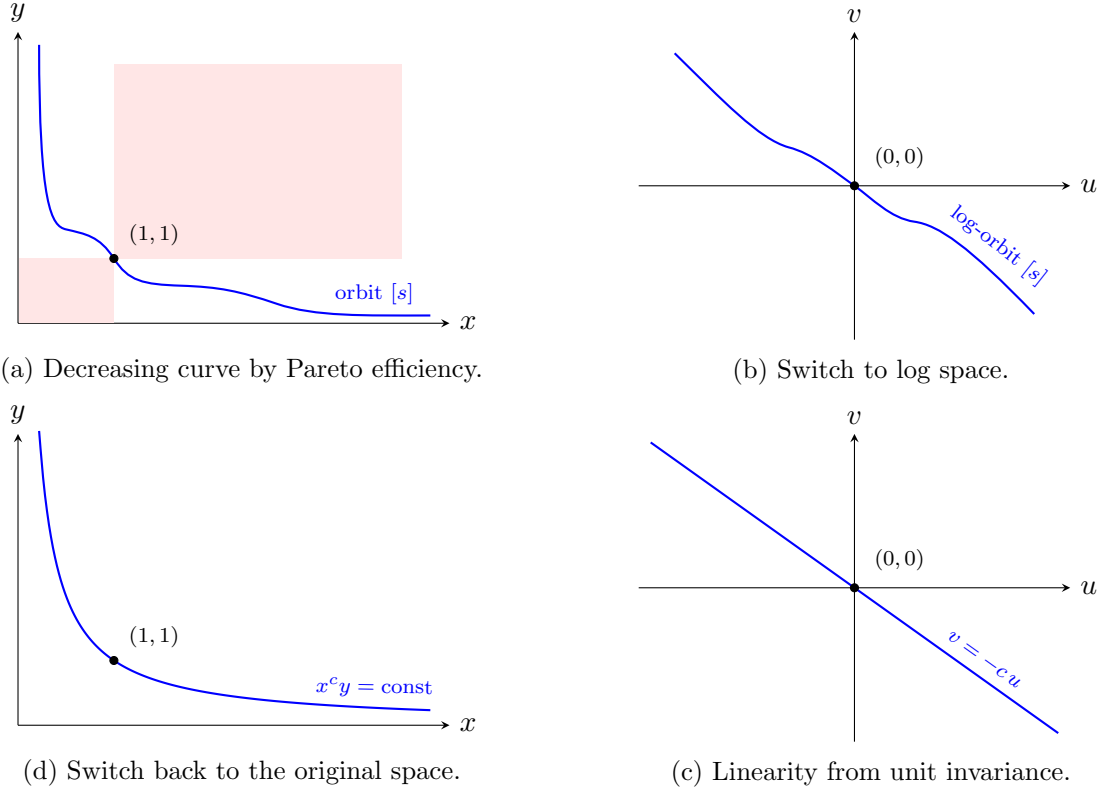
\begin{figure}[H]
  \centering
  \newcommand{\panelw}{0.48\linewidth}
  \newcommand{\panelh}{0.75}% height/width ratio
  %% Row 1: reserve space orbit | log space orbit
  \begin{subfigure}[t]{\panelw}
    \centering
    \begin{tikzpicture}[baseline=(current axis.north)]
      \begin{axis}[
        width=0.95\linewidth, height=\panelh*0.95\linewidth,
        axis lines=middle, xlabel={$x$}, ylabel={$y$},
        xmin=0, xmax=4.5, ymin=0, ymax=4.5,
        xtick=\empty, ytick=\empty, clip=false,
        every axis x label/.style={at={(ticklabel* cs:1)}, anchor=west},
        every axis y label/.style={at={(ticklabel* cs:1)}, anchor=south},
      ]
        \fill[red!10] (1,1) rectangle (4,4);
        \fill[red!10] (0.01,0.01) rectangle (1,1);
        \draw[blue, thick]
          (axis cs:0.22,4.3)
          .. controls (axis cs:0.22,1.8) and (axis cs:0.35,1.5) ..
          (axis cs:0.5,1.45)
          .. controls (axis cs:0.7,1.38) and (axis cs:0.85,1.35) ..
          (axis cs:1.0,1.0)
          .. controls (axis cs:1.15,0.65) and (axis cs:1.3,0.6) ..
          (axis cs:1.7,0.58)
          .. controls (axis cs:2.1,0.56) and (axis cs:2.3,0.5) ..
          (axis cs:2.7,0.3)
          .. controls (axis cs:3.1,0.1) and (axis cs:3.6,0.12) ..
          (axis cs:4.3,0.12);
        \addplot[only marks, mark=*, mark size=1.5pt, black] coordinates {(1,1)};
        \node[anchor=south west, font=\scriptsize] at (1.05,1.05) {$(1,1)$};
        \node[blue, font=\scriptsize, anchor=west] at (3.2,0.45) {orbit $[s]$};
      \end{axis}
    \end{tikzpicture}
    \caption{Decreasing curve by Pareto efficiency.}
    \label{fig:pareto_curve}
  \end{subfigure}
  \hfill
  \begin{subfigure}[t]{\panelw}
    \centering
    \begin{tikzpicture}[baseline=(current axis.north)]
      \begin{axis}[
        width=0.95\linewidth, height=0.78*0.95\linewidth,
        axis lines=middle, xlabel={$u$}, ylabel={$v$},
        xmin=-1.8, xmax=1.8, ymin=-1.8, ymax=1.8,
        xtick={0}, ytick={0}, clip=false, axis on top,
        every axis x label/.style={at={(ticklabel* cs:1)}, anchor=west},
        every axis y label/.style={at={(ticklabel* cs:1)}, anchor=south},
      ]
        \draw[blue, thick]
          (axis cs:-1.5,1.55)
          .. controls (axis cs:-1.1,1.0) and (axis cs:-0.8,0.55) ..
          (axis cs:-0.55,0.45)
          .. controls (axis cs:-0.35,0.38) and (axis cs:-0.15,0.15) ..
          (axis cs:0.0,0.0)
          .. controls (axis cs:0.15,-0.15) and (axis cs:0.3,-0.38) ..
          (axis cs:0.5,-0.42)
          .. controls (axis cs:0.7,-0.46) and (axis cs:0.95,-0.7) ..
          (axis cs:1.5,-1.5);
        \addplot[only marks, mark=*, mark size=1.5pt, black] coordinates {(0,0)};
        \node[anchor=south west, font=\scriptsize] at (0.08,0.08) {$(0,0)$};
        \node[blue, font=\scriptsize, rotate=-38] at (axis cs:1.2,-0.7) {log-orbit $[s]$};
      \end{axis}
    \end{tikzpicture}
    \caption{Switch to log space.}
    \label{fig:log_curve}
  \end{subfigure}

  \vspace{2pt}

  %% Row 2: (c) = origin orbit line (right), (d) = hyperbola (left)
  %% Source order determines subfigure counter: first subfigure = (c), second = (d)
  %% We define (c) first with phantomcaption, then visually place (d) left and (c) right
  \addtocounter{subfigure}{1}% skip (c) for now, assign (c) to the line below
  \begin{subfigure}[t]{\panelw}
    \centering
    \begin{tikzpicture}[baseline=(current axis.north)]
      \begin{axis}[
        width=0.95\linewidth, height=\panelh*0.95\linewidth,
        axis lines=middle, xlabel={$x$}, ylabel={$y$},
        xmin=0, xmax=4.5, ymin=0, ymax=4.5,
        xtick=\empty, ytick=\empty, clip=false,
        every axis x label/.style={at={(ticklabel* cs:1)}, anchor=west},
        every axis y label/.style={at={(ticklabel* cs:1)}, anchor=south},
      ]
        \addplot[blue, thick, smooth, domain=0.22:4.3, samples=100] {1/x};
        \addplot[only marks, mark=*, mark size=1.5pt, black] coordinates {(1,1)};
        \node[anchor=south west, font=\scriptsize] at (1.05,1.05) {$(1,1)$};
        \node[blue, font=\scriptsize, anchor=west, rotate=-1] at (3.0,0.6) {$x^c y=\mathrm{const}$};
      \end{axis}
    \end{tikzpicture}
    \caption{Switch back to the original space.}
    \label{fig:hyperbolas}
  \end{subfigure}
  \hfill
  \addtocounter{subfigure}{-2}% go back to assign (c) to the line
  \begin{subfigure}[t]{\panelw}
    \centering
    \begin{tikzpicture}[baseline=(current axis.north)]
      \begin{axis}[
        width=0.95\linewidth, height=0.78*0.95\linewidth,
        axis lines=middle, xlabel={$u$}, ylabel={$v$},
        xmin=-1.8, xmax=1.8, ymin=-1.8, ymax=1.8,
        xtick={0}, ytick={0}, clip=false, axis on top,
        every axis x label/.style={at={(ticklabel* cs:1)}, anchor=west},
        every axis y label/.style={at={(ticklabel* cs:1)}, anchor=south},
      ]
        \addplot[blue, thick, domain=-1.7:1.7, samples=2] {-1*x};
        \addplot[only marks, mark=*, mark size=1.5pt, black] coordinates {(0,0)};
        \node[anchor=south west, font=\scriptsize] at (0.08,0.08) {$(0,0)$};
        \draw[draw=none] (axis cs:0.5,-0.5) -- (axis cs:1.5,-1.5)
          node[midway, sloped, above, blue, font=\scriptsize] {$v = -c\,u$};
      \end{axis}
    \end{tikzpicture}
    \caption{Linearity from unit invariance.}
    \label{fig:origin_line}
  \end{subfigure}
  \addtocounter{subfigure}{1}% restore counter to (e) position

  \caption{Derivation of the weighted geometric-mean invariant.}
  \label{fig:derivation}
\end{figure}

\textbf{Step~(a): Pareto efficiency.}
Consider the orbit through $(1,1)$.
Pareto efficiency requires that no state on the orbit dominates another:
if $s_1 = (x_1, y_1)$ and $s_2 = (x_2, y_2)$ lie on the same orbit with $x_1 < x_2$, then $y_1 > y_2$.
The orbit is therefore a strictly decreasing curve (Figure~\ref{fig:pareto_curve}).

\medskip\noindent
\textbf{Step~(b): Switch to log space.}
Unit invariance acts multiplicatively: rescaling $(x,y) \mapsto (ax, by)$.
Define the log map $\ell(x,y) := (\log x,\log y)$;
the state $(1,1)$ maps to the origin $(0,0)$.
In log coordinates, rescaling becomes a translation $(u,v) \mapsto (u + \log a, v + \log b)$, which is much easier to analyze.
The orbit through $(0,0)$ is again a strictly decreasing curve (Figure~\ref{fig:log_curve}).

We lift the orbit relation to log space:
$z \approx z' \iff \ell^{-1}(z)\sim \ell^{-1}(z')$.
Unit invariance makes this relation translation-invariant: if $z\approx z'$, then $z+\delta\approx z'+\delta$ for all $\delta\in\mathbb R^2$
(Lemma~\ref{lem:translationInvariantApprox}).

\medskip\noindent
\textbf{Step~(c): Linearity from unit invariance.}
Let $H := \{h\in\mathbb R^2 \mid h\approx 0\}$ be the log-orbit of the origin.
Translation invariance gives $H$ an additive subgroup structure:
$0 \in H$;
$h \in H \Rightarrow -h \in H$;
$h_1, h_2 \in H \Rightarrow h_1 + h_2 \in H$.
But $H$ is also a strictly decreasing curve by Pareto efficiency.
An additive subgroup compatible with the strict monotonicity imposed by Pareto efficiency must be a straight line through the origin, since any genuine bend would contradict closure under addition; see Appendix~\ref{app:fullproofs} for the precise argument.
We conclude that $H$ is a line $v = -c\,u$ with $c > 0$ (Figure~\ref{fig:origin_line}).

\medskip\noindent
\textbf{Step~(d): Return to reserve space.}
Exponentiating $v = -cu$ gives $\log y = -c \log x$, i.e., $x^c y = \mathrm{const}$ (Figure~\ref{fig:hyperbolas}).
Reparameterizing as $w=\frac{c}{1+c}\in(0,1)$ yields the weighted geometric mean $x^w y^{1-w}=\mathrm{const}$.

\subsection{Extend to all orbits}

\begin{figure}[H]
  \centering
  \begin{tikzpicture}
    \begin{axis}[
      width=0.65\linewidth,
      height=0.5\linewidth,
      axis lines=middle,
      xlabel={$u$},
      ylabel={$v$},
      xmin=-1.8, xmax=1.8,
      ymin=-1.8, ymax=1.8,
      xtick={0},
      ytick={0},
      clip=false,
      every axis x label/.style={at={(ticklabel* cs:1)}, anchor=west},
      every axis y label/.style={at={(ticklabel* cs:1)}, anchor=south},
    ]
      % Origin orbit H
      \addplot[blue, thick, domain=-1.7:1.7, samples=2] {-1*x};
      % Another orbit J (parallel, offset +0.7)
      \addplot[red!70!black, thick, domain=-1.7:1.7, samples=2] {-1*x + 0.7};

      % Origin
      \addplot[only marks, mark=*, mark size=2.5pt, black] coordinates {(0,0)};
      \node[anchor=south west, font=\small] at (0.05,0.05) {$0$};

      % Points z1, z2 on J: z1=(-0.5, 1.2), z2=(0.7, 0.0)
      \addplot[only marks, mark=*, mark size=2.5pt, red!70!black] coordinates {(-0.5,1.2) (0.7,0.0)};
      \node[red!70!black, font=\small, anchor=south] at (-0.5,1.25) {$z_1$};
      \node[red!70!black, font=\small, anchor=west] at (0.75,0.0) {$z_2$};

      % z2 - z1 = (1.2, -1.2) on H
      \addplot[only marks, mark=*, mark size=2.5pt, blue] coordinates {(1.2,-1.2)};
      \node[blue, font=\small, anchor=north west] at (1.2,-1.25) {$z_2 - z_1$};

      % Translation arrow from z1 to origin (dashed)
      \draw[->, dashed, gray, thick] (axis cs:-0.5,1.2) -- (axis cs:0.03,0.03);
      \node[gray, font=\scriptsize, anchor=south east] at (-0.15,0.7) {$-z_1$};

      % Translation arrow from z2 to z2-z1 (dashed)
      \draw[->, dashed, gray, thick] (axis cs:0.7,0.0) -- (axis cs:1.17,-1.17);

      % Labels for H and J
      \draw[draw=none] (axis cs:-1.2,1.2) -- (axis cs:-0.2,0.2)
        node[midway, sloped, above, blue, font=\small] {$H$};
      \draw[draw=none] (axis cs:-1.2,1.9) -- (axis cs:-0.2,0.9)
        node[midway, sloped, above, red!70!black, font=\small] {$J$};
    \end{axis}
  \end{tikzpicture}
  \caption{Any orbit $J$ is parallel to $H$.
  Translating by $-z_1$ moves $z_1\in J$ to the origin and $z_2$ to $z_2-z_1$.
  Since translation preserves $\approx$, we get $z_2 - z_1 \in H$,
  so $J$ is a translate of $H$.}
  \label{fig:parallel_lines}
\end{figure}
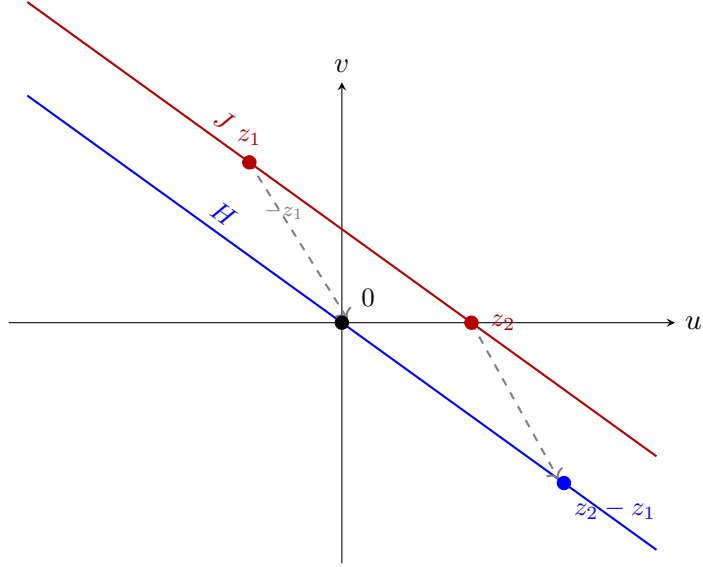

Consider the log orbit $J$ of an arbitrary valid state.
We show that $J$ is a line parallel to $H$.
Take any two points $z_1, z_2 \in J$, so that $z_1 \approx z_2$.
By Lemma~\ref{lem:translationInvariantApprox}, translating the entire space by $-z_1$ shifts $z_1$ to the origin and $z_2$ to $z_2 - z_1$.
Since translation preserves the orbit relation, $0 \approx (z_2 - z_1)$, which means $(z_2 - z_1) \in H$.
Thus, the difference between any two points in $J$ always lies in $H$.
Because $H$ is a line through the origin, $J$ must be a parallel translate of $H$: $J = z_0 + H$ for any fixed $z_0 \in J$.

Applying this argument to every orbit reveals a clean geometric picture: all trading orbits in log space are parallel lines with the same slope $-c$ (Figure~\ref{fig:parallel_lines}).
Different orbits differ only by their intercept.
Exponentiating back to reserve space, each line $cu + v = \mathrm{const}$ becomes $x^c y = \mathrm{const}$, and reparameterizing as $w = c/(1+c)$ gives $x^w y^{1-w} = \mathrm{const}$ on every orbit.

\subsection{Role of the Axioms}

The three axioms play distinct roles in the argument. Validity invariance keeps the orbit relation internal to the positive state space, so that trading orbits are genuinely geometric objects in $\mathcal S$ and can be transported to log space. Pareto efficiency forces each orbit to be a strictly decreasing tradeoff curve rather than a higher-dimensional region; without it, one can have mechanisms whose symmetrized trading orbits fill the entire positive quadrant. Unit invariance is what rigidifies these decreasing curves into weighted geometric-mean level sets: in log coordinates it becomes translation invariance, which turns the orbit through $(1,1)$ into an additive subgroup of $\mathbb R^2$, and Pareto efficiency then forces that subgroup to be a line. Thus Pareto efficiency and unit invariance drive the classification, while validity invariance ensures that the classification is taking place on the intended positive-reserve state space.

\subsection{Derivation of Constant Product}

Token symmetry is the principle that relabeling the two assets should not change
the mechanism.

\begin{assumption}[Token relabeling symmetry]\label{ass:tokenSymmetry}
Let $\tau(x,y)=(y,x)$. For all valid states $s,t$,
\[
s\sim t \quad\Longleftrightarrow\quad \tau(s)\sim \tau(t).
\]
\end{assumption}

\begin{corollary}[Constant Product]\label{cor:constantProductDirect}
Assume the hypotheses of Theorem~\ref{thm:main} and token relabeling
symmetry (Assumption~\ref{ass:tokenSymmetry}). Then $w=\tfrac12$ and every orbit
satisfies
\[
x y=\mathrm{const}.
\]
\end{corollary}
\begin{proof}
In log space, $\tau$ acts as coordinate swap $(u,v)\mapsto(v,u)$.
A line $v=-c u + d$ is sent to $u=-c v + d$, equivalently $v=-(1/c)u + d/c$.
Token symmetry requires the orbit partition to be invariant under this swap,
so the slope parameter must satisfy $c=1/c$, hence $c=1$.
Therefore $w=\frac{c}{1+c}=\frac12$, and $\Phi(x,y)=\sqrt{xy}$, so $xy$ is
constant on orbits.
\end{proof}

\subsection{Extension to Multiple Assets}\label{sec:multiasset}

The two-asset classification extends to pools with arbitrarily many tokens.
Given a multi-asset AMM, restricting any of its global trading orbits to a
coordinate $2$-slice (fixing all reserves except a chosen pair $(i,j)$)
yields a \emph{two-token restriction}, i.e.\ the intersection of the
global orbit with the slice.
The global orbit structure is determined by local conditions on these
restrictions: if every two-token restriction of the AMM satisfies the
two-asset axioms, then a single invariant $\prod x_i^{w_i}$ governs all
trades, and conversely.
We now make this precise.

A state is a vector $s=(x_1,\ldots,x_n)\in\mathbb{R}^n$ with valid
state space
$\mathcal{S}=\{s\in\mathbb{R}^n : x_i>0 \text{ for all } i\}$.
For each pair of distinct tokens $(i,j)$, a swap primitive
$X_{ij}(s,dx)$ inputs $dx\ge 0$ units of token~$i$ and returns
token~$j$:
\[
  (x_i,\,x_j)
  \;\mapsto\;
  \bigl(x_i + dx,\; x_j - y^{ij}_{\mathrm{out}}(s,dx)\bigr),
\]
with all other reserves unchanged.
Adjacency and orbits extend from the two-asset case
(Definitions~\ref{def:adjacency}--\ref{def:orbit}), with the orbit
relation generated by all pairwise swaps.

\begin{definition}[Coordinate two-token slice]
For a valid state $p\in\mathcal S$ and a pair $1\le i<j\le n$, define
\[
  \Pi_{ij}(p):=\{x\in\mathbb R^n : x_k=p_k \text{ for all } k\neq i,j\}.
\]
\end{definition}

\begin{definition}[Two-token $(i,j)$ restriction]
For a valid state $p\in\mathcal S$ and a pair $1\le i<j\le n$, the
\emph{two-token $(i,j)$ restriction through $p$} is the two-asset AMM
on the slice $\Pi_{ij}(p)$ whose swap primitives are the $(i,j)$-swaps
$X_{ij}$ and $X_{ji}$ of the global AMM. Its trading orbit through $p$ is
\[
  [p]_{ij}:=[p]\cap \Pi_{ij}(p).
\]
\end{definition}

\begin{theorem}[Multi-asset orbit classification]\label{thm:multiasset}
The trading orbits of an $n$-asset AMM are exactly the level sets of
a weighted product
\[
  \Phi(x_1,\ldots,x_n)=\prod_{i=1}^{n} x_i^{w_i},
  \qquad w_i>0,\quad \textstyle\sum_{i} w_i=1,
\]
if and only if, for every valid state $p\in\mathcal S$ and every pair
$1\le i<j\le n$, the two-token restriction $[p]_{ij}$, viewed in the
coordinates $(x_i,x_j)$, satisfies the two-asset axioms
(Assumptions~\ref{ass:validity}--\ref{ass:unit}).
\end{theorem}
\begin{proof}[Proof sketch]
The idea is to pass to log coordinates $u_i=\log x_i$, where weighted
products become affine equations and the geometry of the orbit slices
becomes linear.

\smallskip\noindent
\textbf{Backward direction.}
If the global orbits are level sets of
$\Phi(x)=\prod_i x_i^{w_i}$, then in log coordinates they are affine
hyperplanes
\[
  w_1u_1+\cdots+w_nu_n=\mathrm{const}.
\]
Intersecting such a hyperplane with a coordinate slice
$\Pi_{ij}(p)$ gives an affine line in the $(u_i,u_j)$-plane with
slope $-w_i/w_j<0$ (since $w_i,w_j>0$). Exponentiating back, each two-token restriction
$[p]_{ij}$ is a weighted geometric-mean curve
$x_i^{w_i}x_j^{w_j}=\mathrm{const}$, so by Theorem~\ref{thm:main} it
satisfies the two-asset axioms.

\smallskip\noindent
\textbf{Forward direction.}
Conversely, assume every two-token restriction $[p]_{ij}$ satisfies the
two-asset axioms. Then by the two-asset classification, every such
slice is a weighted geometric-mean curve in $(x_i,x_j)$, hence a line
of negative slope in log coordinates. Thus every coordinate $2$-slice
of every log-orbit is a negative-slope affine line. The key geometric
fact, proved in Lemma~\ref{lem:slicesToHyperplane}, is that a subset of
$\mathbb R^n$ with this property must itself be an affine hyperplane
$a_1u_1+\cdots+a_nu_n=d$ with all coefficients $a_i>0$.
Setting $w_i=a_i/\!\sum_k a_k$ and exponentiating back, each orbit is
a level set of a weighted product $\prod_i x_i^{w_i}$ with $\sum_i w_i=1$.
Since parallel log-hyperplanes share the same normal, the same weights
work for every orbit.
\end{proof}

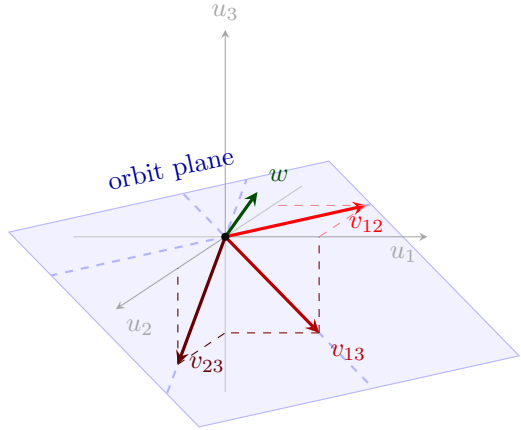
\begin{figure}[H]
  \centering
  \begin{tikzpicture}[
    x={(0.95cm,0cm)},        % u1 axis
    y={(-0.48cm,-0.32cm)},   % u2 axis
    z={(0cm,0.98cm)},        % u3 axis
    >=stealth
  ]
    % --- Orbit plane (shaded) ---
    % The plane is spanned by v12 and v13 from the origin.
    % We'll draw a parallelogram-ish region.
    % Use w=(1,1,1) for equal weights (symmetric case for clarity),
    % so the plane is u1+u2+u3=0, and
    %   v12 = (2/3, -1/3, 0) [rescaled to (2,-1,0)]
    %   v13 = (2/3, 0, -1/3) [rescaled to (2,0,-1)]
    %   v23 = (0, 2/3, -1/3) [rescaled to (0,2,-1)]
    % For visual clarity, use a=1/2 so directions are (1,-1,0), (1,0,-1), (0,1,-1)

    \def\sc{1.55}
    % Plane patch with edge directions parallel to
    % v13=(1,0,-1) and v12=(1,-1,0).
    \coordinate (P1) at (-1.9,{0.95*\sc+0.7},{0.95*\sc-0.7});
    \coordinate (P2) at ({-1.9+1.7*\sc},{0.95*\sc+0.7},{0.95*\sc-0.7-1.7*\sc});
    \coordinate (P3) at ({-1.9+3.6*\sc},{0.95*\sc+0.7-1.9*\sc},{0.95*\sc-0.7-1.7*\sc});
    \coordinate (P4) at ({-1.9+1.9*\sc},{0.95*\sc+0.7-1.9*\sc},{0.95*\sc-0.7});

    \fill[blue!8, opacity=0.7] (P1) -- (P2) -- (P3) -- (P4) -- cycle;
    \draw[blue!30, thin] (P1) -- (P2) -- (P3) -- (P4) -- cycle;

    % --- Axes ---
    \draw[gray!45] (0,0,0) -- (-2.1,0,0);
    \draw[gray!45] (0,0,0) -- (0,-2.1,0);
    \draw[gray!45] (0,0,0) -- (0,0,-2.1);
    \draw[->, gray!70] (0,0,0) -- (2.8,0,0) node[anchor=north east, font=\small]{$u_1$};
    \draw[->, gray!70] (0,0,0) -- (0,3,0) node[anchor=north west, font=\small]{$u_2$};
    \draw[->, gray!70] (0,0,0) -- (0,0,2.8) node[anchor=south, font=\small]{$u_3$};

    % --- Intersections with the coordinate planes ---
    % Draw these as dashed reverse continuations of the swap directions,
    % truncated at the visible plane patch.
    \begin{scope}
      \clip (P1) -- (P2) -- (P3) -- (P4) -- cycle;
      % (u1,u2)-plane: u3=0
      \draw[blue!30, dashed, line width=0.8pt] (-3.2,3.2,0) -- (3.2,-3.2,0);
      % (u1,u3)-plane: u2=0
      \draw[blue!30, dashed, line width=0.8pt] (-3.2,0,3.2) -- (3.2,0,-3.2);
      % (u2,u3)-plane: u1=0
      \draw[blue!30, dashed, line width=0.8pt] (0,-3.2,3.2) -- (0,3.2,-3.2);
    \end{scope}

    % --- Swap direction vectors ---
    % v12: changes u1 up, u2 down (direction (1,-1,0))
    \draw[->, red, very thick] (0,0,0) -- (1.3,-1.3,0)
      node[anchor=north, font=\small]{$v_{12}$};
    % v13: changes u1 up, u3 down (direction (1,0,-1))
    \draw[->, red!70!black, very thick] (0,0,0) -- (1.3,0,-1.3)
      node[anchor=north west, font=\small]{$v_{13}$};
    % v23: changes u2 up, u3 down (direction (0,1,-1))
    \draw[->, red!40!black, very thick] (0,0,0) -- (0,1.3,-1.3)
      node[anchor=west, font=\small]{$v_{23}$};

    % --- Normal vector w ---
    \draw[->, black!70!green, very thick] (0,0,0) -- (0.9,0.9,0.9)
      node[anchor=south west, font=\small]{$w$};

    % --- Dashed projection lines showing each vector's coordinate plane ---
    % v12 = (1.3, -1.3, 0) lies in (u1, u2) plane
    \draw[red!50, dashed, thin] (1.3,-1.3,0) -- (1.3,0,0);
    \draw[red!50, dashed, thin] (1.3,-1.3,0) -- (0,-1.3,0);
    % v13 = (1.3, 0, -1.3) lies in (u1, u3) plane
    \draw[red!50!black, dashed, thin] (1.3,0,-1.3) -- (1.3,0,0);
    \draw[red!50!black, dashed, thin] (1.3,0,-1.3) -- (0,0,-1.3);
    % v23 = (0, 1.3, -1.3) lies in (u2, u3) plane
    \draw[red!30!black, dashed, thin] (0,1.3,-1.3) -- (0,1.3,0);
    \draw[red!30!black, dashed, thin] (0,1.3,-1.3) -- (0,0,-1.3);

    % --- Origin dot ---
    \fill[black] (0,0,0) circle (1.5pt);

    % --- Labels ---
    \node[blue!60!black, font=\small, anchor=south, rotate=12] at (0.15,1.65,1.15)
      {orbit plane};
  \end{tikzpicture}
  \caption{Three-asset AMM in log space, showing the orbit through
    the origin (corresponding to the state $(1,1,1)$).
    The pairwise swap directions $v_{12},v_{13},v_{23}$ (red) all lie
    in a common plane.
    The weight vector $w$ (green) is normal to this plane.
    Every other orbit is a parallel translate of this plane,
    characterized by $\sum w_i u_i = \mathrm{const}$.}
  \label{fig:multiasset_logspace}
\end{figure}

\begin{corollary}[Multi-asset constant product]\label{cor:multiassetCP}
Under the hypotheses of Theorem~\ref{thm:multiasset}, if additionally
the orbit structure is invariant under every permutation of tokens,
then $w_i = 1/n$ for all~$i$ and the orbits satisfy
$\prod_{i=1}^n x_i = \mathrm{const}$.
\end{corollary}

In practice, the invariant $\Phi=\prod x_i^{w_i}$ defines admissible
trades: a state transition from $s$ to~$t$ is a valid fee-free trade
whenever $\Phi(t)=\Phi(s)$.
This includes pairwise swaps, but also multi-token trades such as
depositing several tokens and withdrawing one, or rebalancing across
multiple tokens in a single transaction.
A trade routed through an intermediate token (e.g., swapping $X_1$ for
$X_2$ via~$X_3$) is a chain of pairwise swaps and therefore also
preserves~$\Phi$.

%%%%%%%%%%%%%%%%%%%%%%%%%%%%%%%%%%%%%%%%%%%%%%%%%%%%%%%%%%%%%%%%%%%%%%%%%%%%%%%%

\section{Extension to Fees and Liquidity Operations}\label{sec:fees_liquidity}

Theorems~\ref{thm:main} and~\ref{thm:multiasset} classify \emph{fee-free} swap orbits.
Real-world AMMs also include trading fees and liquidity operations.
These do not invalidate the geometric picture; rather, the partition of the state space into trading orbits provides a natural reference frame:
swaps move \emph{along} orbits, while fees and liquidity operations move the state \emph{across} them.

\medskip\noindent
\textbf{Fees.}\quad
In many AMMs, trading fees are charged on the input token.
For an $X$-in trade of gross size $dx\ge 0$ with fee rate $\phi\in[0,1)$,
only $(1-\phi)\,dx$ is used for pricing, while the remainder $\phi\,dx$ is retained in the 
pool reserves.
Thus a fee-paying swap can be viewed as a fee-free swap with effective 
input $(1-\phi)\,dx$ along the current orbit, followed by a one-sided reserve injection
\[
(x,y)\mapsto (x+\phi\,dx,\,y),
\]
which moves the state across orbits.
This matches Uniswap~v2, where the output $dy$ satisfies
\[
(x+(1-\phi)\,dx)(y-dy)=xy,
\]
while the full input $dx$ is added to the $X$-reserve.
Note that the order matters: reversing it, by injecting the fee first and then swapping on 
the shifted orbit, yields a different output.

\medskip\noindent
\textbf{Liquidity operations.}\quad
Proportional liquidity addition or removal scales the state as
$(x,y)\mapsto (\lambda x,\lambda y)$ with $\lambda>0$.
Since $\Phi_w(\lambda x,\lambda y) = \lambda\,\Phi_w(x,y)$, such operations move the state between orbits.
In log space this becomes the diagonal translation
$(u,v)\mapsto (u+\log\lambda,\, v+\log\lambda)$.
More generally, an asymmetric reserve scaling
$(x,y)\mapsto(\lambda_x x,\,\lambda_y y)$
changes the invariant by the factor $\lambda_x^{\,w}\lambda_y^{\,1-w}$:
\[
\Phi_w(\lambda_x x,\lambda_y y)
=
\lambda_x^{\,w}\lambda_y^{\,1-w}\,\Phi_w(x,y).
\]
This includes single-sided reserve additions as a special case.

\medskip
In summary: fee-free swaps preserve trading orbits, retained input-token fees induce transverse drift across orbits, 
and proportional liquidity operations produce diagonal translations in log space, hence moving the state across orbits.

%%%%%%%%%%%%%%%%%%%%%%%%%%%%%%%%%%%%%%%%%%%%%%%%%%%%%%%%%%%%%%%%%%%%%%%%%%%%%%%%
%%%%%%%%%%%%%%%%%%%%%%%%%%%%%%%%%%%%%%%%%%%%%%%%%%%%%%%%%%%%%%%%%%%%%%%%%%%%%%%%
\section{Lean Formalization}
\label{app:lean_v2}

The results of this paper have been formalized and machine-checked
in Lean~4, using Mathlib for the underlying real analysis.
The source code is available at~\cite{AMMlean}.

For this project, the role of Lean was not only to certify the final
proofs, but also to help shape the paper itself.
Because the paper develops an axiomatic classification theorem, this was
particularly valuable.
The formalization made it transparent which assumptions were genuinely
needed, helped us formulate them more precisely, and provided a
reliable environment for testing variations of the theory.

In an informal development, it is easy to retain assumptions that were
useful in an earlier proof attempt or that simply feel natural.
In Lean, by contrast, every theorem is proved under an explicit list of
hypotheses, so weakening the setup immediately reveals which later steps
break and which do not.
This helped us reduce the number of assumptions in the final
presentation.
For example, an earlier version of the project included an additional
path-independence type assumption, but the formalization showed that it
was not needed for the main classification result.
More generally, once a compiling development was in place, it became
easy to experiment with modified assumptions and alternative proof
structures while keeping full control over the dependency structure of
the argument.

The formalization also led to greater precision in the formulation of
the axioms.
Several statements that sound clear informally had to be made fully
explicit once encoded in Lean.
A representative example is the treatment of validity and adjacency.
At an informal level, one might simply say that swaps ``preserve
validity'' and fold validity directly into the definition of adjacency.
In the formal development, however, we adopted a more careful design
choice.
We modeled validity as a separate predicate on states, while adjacency
records only the existence of a primitive swap.
This made clear that the validity assumption is best formulated in a
way that guarantees that validity is preserved along swap chains,
rather than being left implicit in the definition of the orbit relation.

\section{Conclusion}

The main takeaway is that a minimal set of natural assumptions already
forces a very rigid orbit geometry. In the two-asset case, validity
invariance, Pareto efficiency, and unit invariance force the trading
orbits of a fee-free AMM to be weighted geometric-mean curves
$x^w y^{1-w}=\mathrm{const}$, with constant product recovered under
token symmetry.

The same perspective extends to the multi-asset setting: once every
two-token restriction satisfies the two-asset axioms, the global
trading orbits are forced to be level sets of a weighted product
$\prod_i x_i^{w_i}$ with positive weights $w_i$ summing to 1.
Beyond these classification results themselves,
the Lean formalization was valuable in identifying a minimal set of
assumptions and in guiding the argument toward simpler and more robust
proofs.

%%%%%%%%%%%%%%%%%%%%%%%%%%%%%
\bibliographystyle{plainurl}
\bibliography{references}

\clearpage
%%%%%%%%%%%%%%%%%%%%%%%%%%%%%%%%%%%%%%%%%%%%%%%%%%%%%%%%%%%%%%%%%%%%%%%%%%%%%%%%
\appendix

\section{Full Mathematical Proofs}
\label{app:fullproofs}

\subsection{Basic properties of orbits}

\begin{lemma}[Orbit relation is an equivalence relation]\label{lem:orbitEquiv}
The relation $\sim$ of Definition~\ref{def:orbit} is an equivalence relation on $\mathbb R^2$.
\end{lemma}
\begin{proof}
Reflexivity holds by the zero-step chain $s\sim s$.
Symmetry holds because adjacency $\leftrightarrow$ is symmetric by Definition~\ref{def:adjacency}.
Transitivity holds by concatenating chains.
\end{proof}

\begin{lemma}[Adjacency preserves validity]\label{lem:adjPreservesValidity}
Under Assumption~\ref{ass:validity}, if $s\leftrightarrow t$, then
\[
s\in\mathcal S \iff t\in\mathcal S.
\]
\end{lemma}
\begin{proof}
By Definition~\ref{def:adjacency}, an adjacency step has one of the forms
$t=X(s,\delta)$, $t=Y(s,\delta)$, $s=X(t,\delta)$, or $s=Y(t,\delta)$ with $\delta\ge 0$.
In each case, the equivalence follows directly from Assumption~\ref{ass:validity}.
\end{proof}

\begin{lemma}[Orbit relation preserves validity]\label{lem:orbitPreservesValidity}
Under Assumption~\ref{ass:validity}, if $s\sim t$, then
\[
s\in\mathcal S \iff t\in\mathcal S.
\]
In particular, if $s\in\mathcal S$, then every state in the trading orbit $[s]$ is valid.
\end{lemma}
\begin{proof}
Apply Lemma~\ref{lem:adjPreservesValidity} along a finite adjacency chain from $s$ to $t$.
\end{proof}

\subsection{Consequences of unit invariance}

\begin{lemma}[Unit scaling maps orbits to orbits]\label{lem:TabMapsOrbits}
Under Assumption~\ref{ass:unit}, if $s\sim t$ and $s,t\in\mathcal S$, then
\[
T_{a,b}s \sim T_{a,b}t
\qquad\text{for all } a,b>0.
\]
\end{lemma}
\begin{proof}
It suffices to check one adjacency step.
If $t=X(s,\delta)$, then Assumption~\ref{ass:unit} gives
\[
T_{a,b}t=T_{a,b}X(s,\delta)=X(T_{a,b}s,a\delta),
\]
so $T_{a,b}s\leftrightarrow T_{a,b}t$.
The case $t=Y(s,\delta)$ is analogous.
If the step is written in reverse, symmetry of $\leftrightarrow$ handles it.
Concatenating along a finite chain yields the claim.
\end{proof}

Define the logarithmic coordinate map
\[
\ell:\mathcal S\to\mathbb R^2,
\qquad
\ell(x,y):=(\log x,\log y).
\]

\begin{lemma}[Unit scaling becomes translation in log space]\label{lem:logTranslation}
For all $a,b>0$ and all $s\in\mathcal S$,
\[
\ell(T_{a,b}s)=\ell(s)+(\log a,\log b).
\]
\end{lemma}
\begin{proof}
Immediate from $\log(ax)=\log x+\log a$ and $\log(by)=\log y+\log b$.
\end{proof}

Since $\ell$ is a bijection $\mathcal S\to\mathbb R^2$, its inverse is
$\ell^{-1}(u,v)=(e^u,e^v)$.
Lift the orbit relation to $\mathbb R^2$ by declaring, for $z,z'\in\mathbb R^2$,
\[
z\approx z'
\quad:\Longleftrightarrow\quad
\ell^{-1}(z)\sim \ell^{-1}(z').
\]

\begin{lemma}[Translation invariance of $\approx$]\label{lem:translationInvariantApprox}
Under Assumption~\ref{ass:unit}, if $z\approx z'$, then
\[
z+\delta \approx z'+\delta
\qquad\text{for all } \delta\in\mathbb R^2.
\]
\end{lemma}
\begin{proof}
Write $\delta=(\alpha,\beta)$ and set $a=e^\alpha$, $b=e^\beta$.
Then $\ell^{-1}(z+\delta)=T_{a,b}\ell^{-1}(z)$ and likewise for $z'$.
By Lemma~\ref{lem:TabMapsOrbits}, orbit-equivalence is preserved under $T_{a,b}$.
\end{proof}

\subsection{Structure of log-orbits}

Let $0=(0,0)\in\mathbb R^2$ and define the $\approx$-orbit of the origin,
\[
H := \{h\in\mathbb R^2 \mid h\approx 0\}.
\]
Translation invariance implies that every lifted orbit is a coset of $H$.

\begin{lemma}[$H$ is an additive subgroup]\label{lem:Hsubgroup}
$H$ is an additive subgroup of $(\mathbb R^2,+)$.
\end{lemma}
\begin{proof}
We have $0\in H$ by reflexivity.
If $h\in H$, then $h\approx 0$, and translating by $-h$ yields $0\approx -h$, hence $-h\in H$.
If $h_1,h_2\in H$, then $h_1\approx 0$ and $h_2\approx 0$.
Translating $h_1\approx 0$ by $h_2$ gives $h_1+h_2\approx h_2$, and since $h_2\approx 0$, transitivity yields
$h_1+h_2\approx 0$.
Thus $h_1+h_2\in H$.
\end{proof}

\begin{lemma}[Every lifted orbit is a coset of $H$]\label{lem:cosetsOfH}
For $z,z'\in\mathbb R^2$,
\[
z\approx z'
\quad\Longleftrightarrow\quad
z-z'\in H.
\]
Hence each $\approx$-orbit has the form $z_0+H$.
\end{lemma}
\begin{proof}
If $z\approx z'$, translate by $-z'$ to get $z-z'\approx 0$, so $z-z'\in H$.
Conversely, if $z-z'\in H$, then $z-z'\approx 0$; translating by $z'$ gives $z\approx z'$.
\end{proof}

\begin{lemma}[$H$ is the graph of a strictly decreasing partial function]\label{lem:Hgraph}
For each $u\in\mathbb R$ there is at most one $v\in\mathbb R$ such that $(u,v)\in H$.
Moreover, if $(u_1,v_1),(u_2,v_2)\in H$ and $u_1<u_2$, then $v_1>v_2$.
\end{lemma}
\begin{proof}
Since $(u,v)\in H$ means $\ell^{-1}(u,v)\sim (1,1)$, the claim is exactly
Theorem~\ref{thm:orbitsDecreasing} applied to the orbit of $(1,1)$.
\end{proof}

\begin{lemma}[$H$ has full domain]\label{lem:HfullDomain}
For every $u\in\mathbb R$ there exists $v\in\mathbb R$ such that $(u,v)\in H$.
\end{lemma}
\begin{proof}
First assume $u\ge 0$.
Set
\[
dx:=e^u-1\ge 0,
\qquad
t:=X((1,1),dx).
\]
By Assumption~\ref{ass:validity}, the state $t$ is valid.
Its first coordinate is
\[
t_x = 1 + (e^u-1) = e^u.
\]
Hence
\[
\ell(t)=(u,\log t_y),
\]
and since $t$ is obtained from $(1,1)$ by a single swap, we have $\ell(t)\in H$.

Now let $u<0$.
By the positive case, there exists $v$ such that $(-u,v)\in H$.
Since $H$ is a subgroup, also $(u,-v)=-( -u,v)\in H$.
\end{proof}

\subsection{Consequences of Pareto efficiency}

\begin{theorem}[Orbits are strictly decreasing graphs]\label{thm:orbitsDecreasing}
Fix a valid state $s\in\mathcal S$.
Then there exists a unique function $f_s:\mathrm{proj}_x([s])\to\mathbb R$ such that
\[
(x,y)\in [s] \quad\Longleftrightarrow\quad y=f_s(x),
\]
and $f_s$ is strictly decreasing on its domain.
\end{theorem}
\begin{proof}
Uniqueness at fixed $x$: if $(x,y_1),(x,y_2)\in[s]$ and $y_2\ge y_1$, then
$(x,y_2)\succeq(x,y_1)$.
By Assumption~\ref{ass:pareto}, they must be equal, so $y_1=y_2$.

Strict decrease: if $(x_1,y_1),(x_2,y_2)\in[s]$ with $x_1<x_2$ and $y_2\ge y_1$, then
$(x_2,y_2)\succeq(x_1,y_1)$, again contradicting Assumption~\ref{ass:pareto}.
\end{proof}

\begin{lemma}[$H$ lies in the anti-diagonal region]\label{lem:HantiDiagonal}
Every element of $H$ satisfies
\[
uv\le 0.
\]
Equivalently,
\[
H\subseteq \{(u,v)\in\mathbb R^2 : uv\le 0\}.
\]
\end{lemma}
\begin{proof}
Let $(u,v)\in H$.
Then $(e^u,e^v)\sim(1,1)$.

If $u>0$ and $v>0$, then $(e^u,e^v)\succ(1,1)$, contradicting Assumption~\ref{ass:pareto}.
If $u<0$ and $v<0$, then $(1,1)\succ(e^u,e^v)$, again contradicting Assumption~\ref{ass:pareto}.
Hence the coordinates cannot have the same nonzero sign, so $uv\le 0$.
\end{proof}

\begin{lemma}[Subgroups in the anti-diagonal region are linear]\label{lem:subgroupLinear}
Let $G$ be a nontrivial additive subgroup of $\mathbb R^2$ such that
\[
G\subseteq \{(u,v)\in\mathbb R^2 : uv\le 0\}.
\]
Assume moreover that for every $u\in\mathbb R$ there exists $v\in\mathbb R$ with $(u,v)\in G$, and that for each fixed $u$ there is at most one such $v$.
Then there exists $c>0$ such that
\[
G=\{(u,v)\in\mathbb R^2 : v=-c\,u\}.
\]
\end{lemma}
\begin{proof}
We first show that $G$ contains no two linearly independent vectors.

Suppose for contradiction that $h_1=(u_1,v_1)$ and $h_2=(u_2,v_2)$ are linearly independent elements of $G$.
Since $G=-G$, replacing vectors by their negatives if necessary, we may assume both lie in the closed fourth quadrant, so
\[
u_1,u_2\ge 0, \qquad v_1,v_2\le 0.
\]

If one vector lies on an axis, say $h_1=(u_1,0)$ with $u_1>0$, then
$h_2$ must have $v_2<0$ (otherwise the two vectors are linearly dependent).
For $n$ large enough, $nh_1-h_2=(nu_1-u_2,\,-v_2)$ has both coordinates positive,
contradicting $G\subseteq\{uv\le 0\}$.
Hence we may assume
\[
u_1,u_2>0, \qquad v_1,v_2<0.
\]

Because $h_1$ and $h_2$ are linearly independent, the ratios $|v_1|/|v_2|$ and $u_1/u_2$ are unequal.
After relabeling, assume
\[
\frac{|v_1|}{|v_2|}<\frac{u_1}{u_2}.
\]
Choose $n\in\mathbb N$ so large that the interval
\[
\left(n\frac{|v_1|}{|v_2|},\, n\frac{u_1}{u_2}\right)
\]
contains an integer $m$.
Then for
\[
z:=n h_1-m h_2\in G
\]
we have
\[
z_1=nu_1-mu_2>0,
\qquad
z_2=nv_1-mv_2=-n|v_1|+m|v_2|>0.
\]
Thus $z$ lies in the open first quadrant, contradicting $G\subseteq\{uv\le 0\}$.

Therefore $G$ contains no two linearly independent vectors, so $\operatorname{span}(G)$ is one-dimensional.
Hence $G$ is contained in a line $L$ through the origin.

By the full-domain assumption, there exists some element of $G$ with nonzero first coordinate, so $L$ is not vertical.
Therefore $L$ has the form
\[
L=\{(u,v)\in\mathbb R^2 : v=\kappa u\}
\]
for some $\kappa\in\mathbb R$.
Since for every $u\in\mathbb R$ there exists $(u,v)\in G$ and $G\subseteq L$, the uniqueness assumption forces $v=\kappa u$.
Hence $G=L$.

Finally, because every nonzero element of $G$ satisfies $uv\le 0$, the slope must be non-positive.
It cannot be zero, because then $G$ would be the $u$-axis, contradicting the existence of a unique $v$ for every $u$ together with nontriviality and the anti-diagonal constraint.
So $\kappa<0$.
Writing $\kappa=-c$ with $c>0$ gives
\[
G=\{(u,v)\in\mathbb R^2 : v=-c\,u\}.
\]
\end{proof}

\begin{theorem}[The subgroup $H$ is a line of negative slope]\label{thm:Hline}
There exists $c>0$ such that
\[
H=\{(u,v)\in\mathbb R^2 : v=-c\,u\}.
\]
\end{theorem}
\begin{proof}
Apply Lemma~\ref{lem:subgroupLinear} to $G=H$.

The subgroup property is Lemma~\ref{lem:Hsubgroup}.
The anti-diagonal property is Lemma~\ref{lem:HantiDiagonal}.
Existence of some $v$ for each $u$ is Lemma~\ref{lem:HfullDomain}.
Uniqueness of $v$ at fixed $u$ is Lemma~\ref{lem:Hgraph}.
\end{proof}

\begin{theorem}[Log-orbits are parallel lines with negative slope]\label{thm:logLines}
There exists a constant $c>0$ such that every $\approx$-orbit is a translate of the line
\[
\{(u,v)\in\mathbb R^2 : v=-c\,u\}.
\]
Equivalently, the linear functional
\[
L(u,v):=c\,u+v
\]
is constant on every lifted orbit.
\end{theorem}
\begin{proof}
By Theorem~\ref{thm:Hline}, the orbit of the origin is the line
\[
H=\{(u,v):v=-c\,u\}
\]
for some $c>0$.
By Lemma~\ref{lem:cosetsOfH}, every lifted orbit is a coset of $H$, hence an affine line parallel to it.
On such a coset the functional $L(u,v)=cu+v$ is constant.
\end{proof}
%%%%%%%%%%%%%%%%%%%%%%%%%%%%%%%%%
\subsection{Proof of the main theorem and symmetry corollary}

\begin{theorem}[Weighted-product characterization of orbits]\label{thm:weightedOrbitIff}
There exist constants $c>0$ and $w\in(0,1)$, with
\[
w=\frac{c}{1+c},
\]
such that for all valid states $s=(x_1,y_1)$ and $t=(x_2,y_2)$,
\[
s\sim t
\quad\Longleftrightarrow\quad
x_1^w y_1^{1-w}=x_2^w y_2^{1-w}.
\]
Equivalently, the trading orbits are exactly the level sets of the weighted geometric mean
\[
\Phi_w(x,y):=x^w y^{1-w}.
\]
\end{theorem}
\begin{proof}
By Theorem~\ref{thm:logLines}, there exists $c>0$ such that every lifted orbit in log space is a translate of the line
\[
H=\{(u,v)\in\mathbb R^2 : v=-c\,u\}.
\]
Equivalently, if we define
\[
L(u,v):=c\,u+v,
\]
then two points of $\mathbb R^2$ lie in the same lifted orbit if and only if they have the same $L$-value.

Let $s=(x_1,y_1)$ and $t=(x_2,y_2)$ be valid states, and write
\[
\ell(s)=(\log x_1,\log y_1), \qquad \ell(t)=(\log x_2,\log y_2).
\]
Then
\[
s\sim t
\quad\Longleftrightarrow\quad
\ell(s)\approx \ell(t)
\quad\Longleftrightarrow\quad
L(\ell(s))=L(\ell(t)).
\]
Since
\[
L(\ell(s)) = c\log x_1+\log y_1 = \log(x_1^c y_1),
\]
and similarly for $t$, this is equivalent to
\[
x_1^c y_1 = x_2^c y_2.
\]
Now set
\[
w:=\frac{c}{1+c}\in(0,1).
\]
Because all quantities are positive, raising both sides to the power $1/(1+c)$ preserves equivalence, so
\[
x_1^c y_1 = x_2^c y_2
\quad\Longleftrightarrow\quad
x_1^{c/(1+c)} y_1^{1/(1+c)}
=
x_2^{c/(1+c)} y_2^{1/(1+c)}.
\]
That is,
\[
x_1^w y_1^{1-w}=x_2^w y_2^{1-w}.
\]
Thus
\[
s\sim t
\quad\Longleftrightarrow\quad
\Phi_w(s)=\Phi_w(t),
\]
so the trading orbits are exactly the level sets of $\Phi_w$.
\end{proof}

\begin{corollary}[Constant-product characterization (Corollary~\ref{cor:main_cp} restated)]
Under Assumption~\ref{ass:tokenSymmetry} (token relabeling symmetry), for all valid states $s=(x_1,y_1)$ and $t=(x_2,y_2)$,
\[
s\sim t
\quad\Longleftrightarrow\quad
x_1 y_1 = x_2 y_2.
\]
Equivalently, the trading orbits are exactly the level sets of
\[
\Phi(x,y):=xy.
\]
\end{corollary}
\begin{proof}
By Theorem~\ref{thm:weightedOrbitIff}, there exists $c>0$ such that lifted orbits are affine lines of slope $-c$, equivalently
\[
s\sim t
\quad\Longleftrightarrow\quad
x_1^w y_1^{1-w}=x_2^w y_2^{1-w},
\qquad
w=\frac{c}{1+c}.
\]

Token relabeling symmetry implies that swapping coordinates preserves the orbit relation.
In log space, swapping coordinates sends a line
\[
v=-c\,u+d
\]
to the line
\[
u=-c\,v+d,
\]
equivalently
\[
v=-\frac{1}{c}u+\frac{d}{c}.
\]
Hence the orbit partition is preserved when slope $-c$ is replaced by slope $-1/c$.
Since the orbit partition is unchanged, we must have
\[
c=\frac1c.
\]
Because $c>0$, it follows that $c=1$, and therefore
\[
w=\frac{c}{1+c}=\frac12.
\]

Substituting into Theorem~\ref{thm:weightedOrbitIff} yields
\[
s\sim t
\quad\Longleftrightarrow\quad
x_1^{1/2}y_1^{1/2}=x_2^{1/2}y_2^{1/2}.
\]
Since all quantities are positive, this is equivalent to
\[
x_1 y_1 = x_2 y_2.
\]
Thus the trading orbits are exactly the level sets of $xy$.
\end{proof}

\subsection{Proof of the multi-asset theorem and symmetry corollary}

\begin{lemma}[Negative-slope slices imply hyperplane]\label{lem:slicesToHyperplane}
Let \(S\subseteq \mathbb R^n\) be nonempty. Assume that for every point \(p\in S\) and every pair \(1\le i<j\le n\), the coordinate \(2\)-slice
\[
S\cap \Pi_{ij}(p),
\qquad
\Pi_{ij}(p):=\{x\in\mathbb R^n : x_k=p_k \text{ for all } k\neq i,j\},
\]
is an affine line of negative slope in the \((x_i,x_j)\)-coordinates.

Then \(S\) is an affine hyperplane of the form
\[
a_1x_1+\cdots+a_nx_n=d
\]
with \(a_1,\dots,a_n>0\).
\end{lemma}

\begin{proof}
Fix any \(p\in S\). We first show that \(S\) is the graph of a function \(f:\mathbb R^{n-1}\to\mathbb R\).

\emph{Existence.} Let \((u_1,\dots,u_{n-1})\in\mathbb R^{n-1}\) be given. Starting from \(p\), we perform \(n-1\) moves: at step \(i\), walk along the \((x_i,x_n)\)-slice of the current point until the \(i\)-th coordinate equals \(u_i\). This is possible because the slice is a line of nonzero slope. The move changes only coordinates \(x_i\) and \(x_n\), so coordinates set in earlier steps are preserved. After \(n-1\) steps, the first \(n-1\) coordinates of the resulting point are \(u_1,\dots,u_{n-1}\). Hence there exists \(u_n\in\mathbb R\) such that
\[
(u_1,\dots,u_{n-1},u_n)\in S.
\]

\emph{Uniqueness.} If
\[
(u_1,\dots,u_{n-1},v_n),\ (u_1,\dots,u_{n-1},w_n)\in S,
\]
they share their \(x_1\)-coordinate and lie on the same \((x_1,x_n)\)-slice; since that slice is a line of nonzero slope, \(v_n=w_n\).

Therefore
\[
S=\{(x_1,\dots,x_{n-1},f(x_1,\dots,x_{n-1})) : (x_1,\dots,x_{n-1})\in\mathbb R^{n-1}\}
\]
for a uniquely determined function \(f:\mathbb R^{n-1}\to\mathbb R\).

Fix any \(i\in\{1,\dots,n-1\}\) and any values \(\bar x_k\in\mathbb R\) for \(k\in\{1,\dots,n-1\}\setminus\{i\}\). The graph of the one-variable restriction \(x_i\mapsto f(\bar x_1,\dots,x_i,\dots,\bar x_{n-1})\) is exactly the \((x_i,x_n)\)-slice of \(S\) through any of its points (all coordinates \(x_k\) with \(k\notin\{i,n\}\) being held at \(\bar x_k\)). By hypothesis this slice is an affine line, hence the restriction is affine. Thus \(f\) is affine in each variable separately, hence multi-affine:
\[
f(x_1,\dots,x_{n-1})=\sum_{I\subseteq\{1,\dots,n-1\}}c_I\prod_{k\in I}x_k,
\]
with \(c_I\in\mathbb R\) and the empty product equal to \(1\).

We claim that all mixed terms vanish. Fix \(1\le i<j\le n-1\) and any values \(\bar x_k\in\mathbb R\) for \(k\in\{1,\dots,n-1\}\setminus\{i,j\}\). The resulting two-variable restriction has the form
\[
g(u,v)=\alpha uv+\beta u+\gamma v+\delta,
\]
where the coefficients \(\alpha,\beta,\gamma,\delta\) depend on the frozen values \(\bar x_k\); in particular
\[
\alpha \;=\; \sum_{T\subseteq\{1,\dots,n-1\}\setminus\{i,j\}}\, c_{\{i,j\}\cup T}\prod_{k\in T}\bar x_k.
\]
Its level sets \(g(u,v)=z_0\) are exactly the corresponding \((x_i,x_j)\)-slices of \(S\). By hypothesis each such level set is an affine line of negative slope. If \(\alpha\neq 0\), then \(g(u,v)=z_0\) is a (possibly degenerate) hyperbola, which cannot be a single affine line; this contradicts the hypothesis. Hence \(\alpha=0\) for every choice of the \(\bar x_k\). Since \(\alpha\) is a polynomial in these variables, all its coefficients vanish, i.e.\ \(c_{\{i,j\}\cup T}=0\) for every \(T\subseteq\{1,\dots,n-1\}\setminus\{i,j\}\).

Since this holds for every pair \(i<j\), every coefficient \(c_I\) with \(|I|\ge 2\) vanishes. Hence \(f\) is affine:
\[
f(x_1,\dots,x_{n-1})=b_1x_1+\cdots+b_{n-1}x_{n-1}+d
\]
for suitable \(b_1,\dots,b_{n-1},d\in\mathbb R\). Therefore
\[
S=\{(x_1,\dots,x_n)\in\mathbb R^n : x_n=b_1x_1+\cdots+b_{n-1}x_{n-1}+d\},
\]
so \(S\) is an affine hyperplane.

Finally, if all coordinates except \(x_i,x_n\) are fixed, then the corresponding slice is given by
\[
x_n=b_i x_i+\text{constant},
\]
hence has slope \(b_i\). By hypothesis this slope is negative, so \(b_i<0\) for all \(i<n\). Rewriting the equation of the hyperplane as
\[
(-b_1)x_1+\cdots+(-b_{n-1})x_{n-1}+x_n=d,
\]
and setting
\[
a_i:=-b_i \quad (1\le i\le n-1), \qquad a_n:=1,
\]
we obtain
\[
S=\{x\in\mathbb R^n : a_1x_1+\cdots+a_nx_n=d\}
\]
with \(a_1,\dots,a_n>0\), as claimed.
\end{proof}

\begin{lemma}[Hyperplane implies negative-slope slices]\label{lem:hyperplaneToSlices}
Let \(S\subseteq \mathbb R^n\) be an affine hyperplane of the form
\[
a_1x_1+\cdots+a_nx_n=d
\]
with \(a_1,\dots,a_n>0\).

Then for every point \(p\in S\) and every pair \(1\le i<j\le n\), the coordinate \(2\)-slice
\[
S\cap \Pi_{ij}(p),
\qquad
\Pi_{ij}(p):=\{x\in\mathbb R^n : x_k=p_k \text{ for all } k\neq i,j\},
\]
is an affine line of negative slope in the \((x_i,x_j)\)-coordinates.
\end{lemma}

\begin{proof}
Fix \(p\in S\) and \(1\le i<j\le n\). On the slice \(\Pi_{ij}(p)\) every coordinate except \(x_i\) and \(x_j\) equals the corresponding coordinate of \(p\), so the hyperplane equation reduces to
\[
a_ix_i+a_jx_j = d-\sum_{k\neq i,j}a_kp_k =: c.
\]
Solving for \(x_j\) gives
\[
x_j = -\frac{a_i}{a_j}\,x_i + \frac{c}{a_j},
\]
which describes an affine line in the \((x_i,x_j)\)-plane with slope \(-a_i/a_j<0\), since \(a_i,a_j>0\).
\end{proof}

\begin{theorem}[Multi-asset orbit classification (Theorem~\ref{thm:multiasset} restated)]
The trading orbits of an $n$-asset AMM are exactly the level sets of
a weighted product
\[
  \Phi(x_1,\ldots,x_n)=\prod_{i=1}^{n} x_i^{w_i},
  \qquad w_i>0,\quad \textstyle\sum_{i} w_i=1,
\]
if and only if, for every valid state $p\in\mathcal S$ and every pair
$1\le i<j\le n$, the two-token restriction $[p]_{ij}$, viewed in the
coordinates $(x_i,x_j)$, satisfies the two-asset axioms
(Assumptions~\ref{ass:validity}--\ref{ass:unit}).
\end{theorem}
\begin{proof}
\textbf{Backward direction.}
Suppose the orbits are level sets of $\Phi=\prod x_i^{w_i}$.
In log coordinates $u_i=\log x_i$, each orbit is the hyperplane
\[
w_1u_1+\cdots+w_nu_n=\mathrm{const}
\]
with all $w_i>0$.
By Lemma~\ref{lem:hyperplaneToSlices}, every coordinate $2$-slice of this hyperplane is a line of negative slope.
Exponentiating back, for every valid state $p$ and every pair $(i,j)$,
the restricted orbit $[p]_{ij}$ is a weighted geometric-mean curve
\[
x_i^{w_i}x_j^{w_j}=\mathrm{const}
\]
in the $(x_i,x_j)$-plane, so it satisfies the two-asset axioms by
Theorem~\ref{thm:main}.

\textbf{Forward direction.}
Suppose that for every valid state $p$ and every pair $(i,j)$, the
two-token restriction $[p]_{ij}$ satisfies the two-asset axioms.
Fix a valid state $s$ and consider its orbit $O$.
Pass to log coordinates $u_i=\log x_i$ and let $S=\ell(O)$ be the corresponding log-orbit.

For any point $p\in S$ and any pair $1\le i<j\le n$, the coordinate $2$-slice $S\cap\Pi_{ij}(p)$ is exactly the log-orbit of the two-token restriction obtained by fixing all reserves except $(x_i,x_j)$ at the values determined by $p$.
By hypothesis this restriction satisfies the two-asset axioms, so by Theorem~\ref{thm:main} its orbit is a level set of $x_i^{a}x_j^{1-a}$ for some $a\in(0,1)$.
In log coordinates this is a line of negative slope.

Since every $2$-slice of $S$ is a line of negative slope, Lemma~\ref{lem:slicesToHyperplane} gives
\[
S=\{u\in\mathbb R^n : a_1u_1+\cdots+a_nu_n=d\}
\]
with all $a_i>0$.
Setting $w_i=a_i/\!\sum_k a_k$ and exponentiating back, the orbit of $s$ is the level set of $\prod x_i^{w_i}$.

Since this holds for every orbit, and the weights are determined by the hyperplane normal (which is the same for all parallel orbits), the trading orbits are exactly the level sets of $\Phi=\prod x_i^{w_i}$ with $w_i>0$ and $\sum w_i=1$.
\end{proof}

\begin{corollary}[Multi-asset constant product (Corollary~\ref{cor:multiassetCP} restated)]
Under the hypotheses of Theorem~\ref{thm:multiasset}, if additionally
the orbit structure is invariant under every permutation of tokens,
then $w_i = 1/n$ for all~$i$ and the orbits satisfy
$\prod_{i=1}^n x_i = \mathrm{const}$.
\end{corollary}
\begin{proof}
By Theorem~\ref{thm:multiasset}, the orbits are level sets of $\prod x_i^{w_i}$ with $w_i>0$ and $\sum w_i=1$.
A permutation $\sigma$ of the tokens sends the invariant $\prod x_i^{w_i}$ to $\prod x_i^{w_{\sigma(i)}}$.
If the orbit structure is invariant under every such permutation, the two invariants must define the same level sets, which forces $w_{\sigma(i)}=w_i$ for all $\sigma$ and all $i$.
Hence all weights are equal, and $\sum w_i=1$ gives $w_i=1/n$.
\end{proof}

\end{document}